\documentclass[preprint,3p,times]{elsarticle}

\usepackage{amssymb}
\usepackage{amsthm}
\newtheorem{proposition}{Proposition}

\usepackage{graphicx}
\usepackage{amsfonts}
\usepackage{bm}
\usepackage{color}
\usepackage{xspace}
\usepackage{array}
\usepackage{pgfplots}
\usepackage{ifthen}
\usepackage{mathtools}
\mathtoolsset{centercolon}
\usepackage{algorithm}
\usepackage{algpseudocode}

\usepackage{xcolor}
\definecolor{mycolor1}{rgb}{0, 0.6039, 0.9804}
\definecolor{mycolor2}{rgb}{0.9294, 0.3686, 0.5765}

\newcommand{\tk}{\tilde{k}}
\newcommand{\tx}{\tilde{x}}
\newcommand{\tSigma}{\tilde{\Sigma}}

\newcommand{\bbR}{\mathbb{R}}
\newcommand{\bbS}{\mathbb{S}}

\newcommand{\calO}{\mathcal{O}}

\newcommand{\calS}{\mathcal{S}}

\newcommand{\rmd}{\mathrm{d}}

\newcommand{\rmN}{\mathrm{N}}

\newcommand{\trans}{{\text{\tiny\sf T}}}

\DeclareMathOperator{\trace}{trace}
\DeclareMathOperator{\proj}{proj}

\usepackage{CJKutf8}

\DeclareMathOperator*{\argmin}{argmin}
\DeclarePairedDelimiter\prn{(}{)}
\DeclarePairedDelimiter\set{\{}{\}}

\usepackage{pgfplots}

\begin{document}

\begin{frontmatter}



\title{Modelling the discretization error of initial value
problems using the Wishart distribution}

\author[label0,label1]{Naoki Marumo}
\author[label1,label2]{Takeru Matsuda}
\author[label3]{Yuto Miyatake\corref{cor1}}
\affiliation[label0]{organization={NTT Communication Science Laboratories},
state={Kyoto},
country={Japan}
}
\affiliation[label1]{organization={Graduate School of Information Science and Technology, The University of Tokyo},
state={Tokyo},
country={Japan}
}
\affiliation[label2]{organization={RIKEN Center for Brain Science},
state={Saitama},
country={Japan}}
\affiliation[label3]{organization={Cybermedia Center, Osaka University},
            state={Osaka},
            country={Japan}
}
\ead{yuto.miyatake.cmc@osaka-u.ac.jp}



\begin{abstract}

This paper presents a new discretization error quantification method for the numerical integration of ordinary differential equations.
The error is modelled by using the Wishart distribution, which enables us to capture the correlation between variables.
Error quantification is achieved by solving an optimization problem under the order constraints for the covariance matrices.
An algorithm for the optimization problem is also established in a slightly broader context.

\end{abstract}



\begin{keyword}
discretization error \sep ODEs \sep Wishart distribution



\end{keyword}

\end{frontmatter}


\section{Introduction}
\label{intro}

In the numerical analysis of differential equations, the error behaviour induced by discretization is of significant importance.
Bounds on the error using constants and parameters, such as time step size, are crucial, and typically obtained theoretically. 
However, demands for quantifying the discretization error and reliability of numerical solutions have surged recently:
sufficiently accurate numerical results are not always achievable, especially for large-scale problems, chaotic systems, and long-time integration;
in the contexts of image processing and machine learning, high accuracy is not necessarily required.
Since overly rough computation is not acceptable for both scenarios, it is vital to evaluate the computational results in a quantitative manner to guarantee or comprehend the reliability of the computation.

Recently, methods quantifying error and reliability using probabilistic or statistical discussions emerged, including ODE filters (and smoothers)~\cite{ks20,tk19,ss19,ts21} and perturbative numerical methods~\cite{ag20,cg17,ls19,ls22}.
These have been studied within a relatively new research area known as probabilistic numerics~\cite{ho22}. 
These algorithms themselves possess varying levels of probabilistic or statistical characteristics.

The present authors have proposed non-probabilistic algorithms for quantification, although they are grounded in certain probabilistic or statistical arguments~\cite{mm21,mm22}. 
These can only be applied in inverse problem settings, but they prove to be quite efficient as the algorithms utilize observation data as prior information. 
However, these methods focus on a single specific variable, neglecting the correlations between variables.
While focusing on a specific variable renders the resulting algorithms efficient, this strong assumption appears overly restrictive.

In this paper, we shall generalize our previous method~\cite{mm21}, which is based on isotonic regression, to handle multiple variables simultaneously. 
The key idea is to model the square of the difference between the numerical approximation and observation at each discrete time using the Wishart distribution, whereas previous studies have used the chi-square distribution. 
We shall also design an algorithm to solve the extended isotonic regression problem efficiently.
Our algorithm resembles the one proposed in~\cite{cd91}. 
However, in~\cite{cd91}, there is some ambiguity in the description of the algorithm. 
To avoid ambiguity, we shall detail the algorithm's derivation. 
We assert that the main novelties lie in modelling of the discretization error using the Wishart distribution, and showing how this modelling and algorithm provide information on the reliability of numerical approximations.
Applications to inverse problems are not discussed in this paper.

\section{Quantifying the numerical error using Wishart distribution}

\subsection{Backgrounds and problem settings}

The specific background of this paper originates from inverse problems.
Consider the initial value problem
\begin{equation}
    \frac{\rmd}{\rmd t} x(t;\theta)
    =
    f(x(t;\theta),\theta), \quad
    x(0;\theta) = x_0 \in V
\end{equation}
with unknown parameters $\theta \in \Theta$, where $V$ describes an appropriate space to which the solution $x(t;\theta)$ belong, and $f:V \times \Theta \to V$ is assumed to be sufficiently regular.
Some variables of the initial state might be included in the unknowns.
For simplicity, we assume no modelling uncertainties, indicating the existence of a true parameter $\theta^\ast \in \Theta$.

Assume that a time series of noisy observations is obtained at $t=t_1,\dots,t_N$ ($0\leq t_1<\dots<t_N$).
The solution operator $\calS_N:\Theta \to V^N$ is defined by $\calS_N(\theta) = [x(t_1;\theta)^\trans,\dots,x(t_N;\theta)^\trans]$.
The observation operator $\calO: V\to\bbR^p$ is assumed to be linear,
and the observation noise is assumed to be a $p$-dimensional Gaussian vector with mean zero and covariance matrix $\Gamma \in \bbR^{p \times p}$.
The observation at $t=t_i$ is denoted by $y_i$:
$y_i = \calO ( x(t_i;\theta^\ast)) + e_i$,
where $e_i \sim \rmN_p (0,\Gamma)$.
This operator is readily generalized to $\calO_i:V^N \to \bbR^d$ such that $\calO_i \circ \calS_N (\theta) = \calO (x(t_i;\theta))$.
The maximum likelihood estimate of $\theta$ is then given by
\begin{equation*}
    \hat{\theta}_{\text{ML}}
    =
    \argmin_{\theta \in \Theta}
    \sum_{n=1}^N (y_i - \calO_i \circ \calS_N (\theta))^\trans \Gamma^{-1}
    (y_i - \calO_i \circ \calS_N (\theta)).
\end{equation*}
However, the true solution map $\calS_N$ is unavailable in general.
Thus, we usually consider the quasi-maximum likelihood estimate of $\theta$ using an approximate solution operator $\tilde{\calS}_N:\Theta \to V^N$:
\begin{equation*}
    \hat{\theta}_{\text{QML}}
    =
    \argmin_{\theta \in \Theta}
    \sum_{n=1}^N (y_i - \calO_i \circ \tilde{\calS}_N (\theta))^\trans  \Gamma^{-1}
    (y_i - \calO_i \circ \tilde{\calS}_N (\theta)).
\end{equation*}
Typically, the approximate operator $\tilde{\calS}_N:\Theta\to V^N$ is defined by $\tilde{\calS}_N(\theta) = [\tx_1(\theta)^\trans,\dots,\tx_N(\theta)^\trans]$, where 
$\tx_i (\theta)$ is a numerical approximation of $x(t_i; \theta)$ obtained by using some numerical integrators such as the Runge--Kutta method.

The quasi-maximum likelihood estimate may have non-negligible bias:
if the approximation $\tilde{S}_N$ is not accurate enough compared with the scale of the observation noise, the bias may be significant and cannot be disregarded (see Example~2.1 of \cite{mm21}).
A potential solution is to introduce a model connecting the observation and numerical approximation: $y_i  = \calO_i \circ \tilde{\calS}_N(\theta^\ast) + \xi_i$,
where 
$\xi_i \sim \rmN_p (0,\Gamma + \Sigma_i)$ and $\Sigma_i$ specifies the scale of the discretization error, that is, $x(t_i;\theta) - x_i(\theta)$.
This model leads to the formulation
\begin{equation*}
    \hat{\theta}
    =
    \argmin_{\theta \in \Theta}
    \sum_{i=1}^N (y_i - \calO_i \circ \tilde{\calS}_N (\theta)) ^\trans (\Gamma+\Sigma_i)^{-1}
    (y_i - \calO_i \circ \tilde{\calS}_N (\theta)).
\end{equation*}
Here, the main idea is that considering the covariance matrix $\Sigma_i$ could potentially yield a less biased estimator and provide uncertainty quantification of the obtained estimate in a more suitable manner~\cite{mm21}. 
We note that a similar approach of adding discretization error as a covariance matrix to the observation model was introduced in the context of Bayesian inverse problems~\cite{co16}.

Estimation of $\Sigma_i$ needs to be addressed.
In the previous papers~\cite{mm21,mm22} by the present authors, both $\Gamma$ and $\Sigma_i$'s are assumed to be diagonal and an iterative method is proposed for estimating $\theta$ and $\Sigma_i$'s.
Starting with an initial guess $\theta^{(0)}$, $\Sigma_i^{(0)}$ is estimated.
Then, with $\Sigma_i^{(0)}$ fixed, $\theta^{(0)}$ is updated to $\theta^{(1)}$, and with $\theta^{(1)}$ fixed, $\Sigma_i^{(0)}$ is updated to $\Sigma_i^{(1)}$.
This iterative procedure continues until some convergence criteria are met.
In estimation of $\Sigma_i$'s, isotonic regression techniques are employed. 
Specifically, $\Sigma_i$'s are updated by solving a certain optimization problem under the constraint that  
each diagonal element of $\Sigma_i$'s is (piecewise) monotonically increasing with respect to $i$.
This constraint reflects the observation that for most problems and numerical integrators, the discretization error gets accumulated over time.
The diagonality assumption makes it possible to solve the optimization problem for $\Sigma_i$'s exactly and efficiently.
However, these approaches overlook correlations between variables.

\subsection{A new model}

In this paper, we propose a new model for the discretization error covariance matrices $\Sigma_i$'s that is free from the diagonality assumption, which can capture (spatial) correlation between discretization error of each variable with the off-diagonal elements of $\Sigma_i$.
For simplicity, we assume the true parameter $\theta^\ast$ is available and focus solely on estimating $\Sigma_i$'s.
We propose a model in which $0 \preceq \Sigma_1 \preceq \cdots \preceq \Sigma_N$, where, for symmetric positive semi-definite matrices $X$ and $Y$, $X\preceq Y$ means that $Y-X$ is symmetric positive semi-definite.
Note that each $\Sigma_i$ is not assumed to be diagonal.
We also propose a methodology for updating $\Sigma_i$ in the next section.

To ensure that the resulting optimization problem is well-defined,
we assume that $\Sigma_i$ is piecewise constant.
Specifically, partitioning $N$ into $n$ parts, we assume that 
\begin{align}
    \label{eq:partition}
    \Sigma_1 = \dots= \Sigma_{k_1}, \quad 
    \Sigma_{k_1+1} = \dots = \Sigma_{k_1 + k_2},\quad 
    \dots,\quad 
    \Sigma_{k_1+\dots + k_{n-1} + 1} = \dots = \Sigma_{k_1+\dots + k_{n}} (= \Sigma_N),
\end{align}
where $k_1+\dots + k_{n} = N$.
We will write $k_1 + \dots + k_i = \tk_{i}$ and $\Sigma_{\tk_i} = \tSigma_i$.
Since $\xi_{\tk_{i-1} +1}, \dots, \xi_{\tk_i} \sim \mathrm{N}_p(0,\Gamma + \tSigma_i)$,
\begin{equation*}
    \sum_{j=1}^{k_i} \xi_{\tk_{i-1} +j} \xi_{\tk_{i-1} +j}^\trans \sim W_p (k_i,\Gamma + \tSigma_i),
\end{equation*}
where $W_p(k,V)$ denotes the Wishart distribution with $k$ degrees of freedom for $p\times p$ matrices. 
This model leads to the following formulation using the new notation $Q_i = \Gamma + \tilde{\Sigma}_i$:
\begin{align}
    \min_{Q \in (\bbS_{+}^p)^n} \sum_{i=1}^n k_i (-\log \det(Q_i^{-1}) + \trace (S_i Q_i^{-1})) 
    \quad  \text{s.t. } \Gamma\preceq Q_1\preceq \cdots \preceq Q_n, \label{form2}
\end{align}
where $\bbS_{+}^p$ is the set of symmetric positive semi-definite matrices of size $p\times p$, $(\bbS_{+}^p)^n$ is its $n$-tuple and
$ S_i = \frac{1}{k_i} \sum_{j=1}^{k_i} \xi_{\tk_{i-1} +j} \xi_{\tk_{i-1} +j}^\trans$.
Here, inside the summation is the negative log-likelihood of observation $S_i$ given the covariance $Q_i$: the likelihood is proportional to $\exp \big( -\frac{k_i}{2} \trace (Q_i^{-1}S_i)\big) / \det (Q_i)^{k_i/2}$.
We assume that the covariance matrix $\Gamma$ is positive definite, and then the matrix $Q_i$ is invertible as long as $Q_i \succeq \Gamma$.

\section{Algorithm}

We develop an algorithm for solving the problem \eqref{form2} in a slightly broader context.

Let $G = (V,E)$ be a directed acyclic graph (DAG)\footnote{Even if $G$ has cycles, we can decompose $G$ into strongly connected components to reduce the problem to the case where $G$ is a DAG.} with vertex set $V \coloneqq \{0,1,\dots,n\}$ and edge set $E$.
We assume that there exists a path from vertex $0 \in V$ to every other vertex.
We are now concerned with the problem
\begin{equation}
    \label{prob:form1}
    \min_{ Q \in (\bbS_{+}^p)^V}
    \sum_{i\in V \setminus \set{0}}
    k_i \prn*{
        -\log \det(Q_i^{-1}) + \trace (S_i Q_i^{-1})
    },
    \quad 
    \text{s.t. }
    Q_0 = \Gamma
    \text{ and }
    Q_i \preceq Q_j
    \text{ for all }
    (i,j)\in E,
\end{equation}
which generalizes the problem \eqref{form2}.

Let us define the functions $f,\iota_{\bbS_+^p}:\bbS^p \to \bbR \cup \{ +\infty\}$ by
\begin{align*}
    f(X) = \begin{cases}
        -\log \det(X) & \text{if } X \succ O,\\
        +\infty & \text{otherwise},
    \end{cases}
    \quad \iota_{\bbS_+^p} (X)
    = \begin{cases}
        0 & \text{if } X \succeq O,\\
        +\infty & \text{otherwise}.
    \end{cases}
\end{align*}
Let $(b_{ie}) \in \bbR^{V \times E}$ be the incidence matrix of $G$: $b_{ie} = 1$ if $e = (i,j)$ for some $j \in V$, $b_{ie} = -1$ if $e = (j,i)$, and $b_{ie} = 0$ otherwise.
The variable transformation $P_i := Q_i^{-1}$ leads to the equivalent form for \eqref{prob:form1}:
\begin{equation}
    \label{prob:form3}
    \min_{ P \in (\bbS^p)^V}
    \sum_{i\in V \setminus \set{0}}
    k_i \prn*{
        f(P_i) + \trace(S_i P_i)
    } 
    + 
    \sum_{e \in E} 
    \iota_{\bbS_+^p} \bigg( \sum_{i \in V} b_{ie} P_i \bigg),
    \quad 
    text{s.t. }
    P_0 = \Gamma^{-1}.
\end{equation}

\subsection{Dual Problem}
The Fenchel dual of the problem \eqref{prob:form3} is 
\begin{equation}
    \label{prob:form4}
    \max_{Y\in (\bbS_+^p)^E} 
    - \sum_{e \in E} \trace \prn*{b_{0e} Y_e \Gamma^{-1}} 
    - \sum_{i\in V \setminus \set{0}}
    k_i f^\ast 
    \bigg(
        \frac{1}{k_i} \sum_{e \in E} b_{ie} Y_e - S_i
    \bigg),
\end{equation}
where $f^\ast \colon \bbS^p \to \bbR \cup \{ +\infty\}$ is the convex conjugate of $f$:
\begin{equation*}
    f^\ast (X) = 
    \begin{cases}
        -\log \det(-X) - p & \text{if } X\prec O, \\
        +\infty & \text{otherwise}.
    \end{cases}
\end{equation*}

One of the Karush--Kuhn--Tucker (KKT) conditions for the problems \eqref{prob:form3} and \eqref{prob:form4} is 
\begin{align*}
    \frac{1}{k_i} \sum_{e \in E} b_{ie} Y_e - S_i
    = \nabla f(P_i)
    = - P_i^{-1}
    \text{ for all }
    i \in V \setminus \set{0}.
\end{align*}
Therefore, given the optimal solution $\hat{Y}$ for the dual problem \eqref{prob:form4}, we can obtain the optimal solution $\hat Q$ for \eqref{prob:form1} by  
\begin{equation}
    \hat Q_i
    =
    S_i - \frac{1}{k_i} \sum_{e \in E} b_{ie} \hat Y_e.
\end{equation}
See \cite[Section~31]{r70} for more mathematical details on the Fenchel dual problem and KKT conditions.

\subsection{Dual block coordinate ascent algorithm}
We propose to apply a block coordinate ascent algorithm to the dual problem \eqref{prob:form4}:
for $e\in E$, all variables $Y_{e^\prime}$ ($e^\prime \neq e$) are fixed, and optimize only for $Y_e$, and conduct this procedure repeatedly while swapping edges until some convergence criteria are met.

Before discussing the optimization method for $Y_e$, we describe how to get a feasible starting point.
Note that $Y \in (\bbS_+^p)^E$ is feasible for the problem~\eqref{prob:form4} (i.e., the objective function value is finite) if and only if 
\begin{align}
    \label{eq:constraint_dual}
    \sum_{e \in E} b_{ie} Y_e
    \prec
    k_i S_i
\end{align}
for all $i \in V \setminus \set{0}$.
Therefore, the solution $Y$ such that $Y_e = O$ for all $e \in E$ is feasible if $S_i \succ O$ for all $i \in V \setminus \set{0}$.
Otherwise, there exists $i' \in V \setminus \set{0}$ that violates the constraint~\eqref{eq:constraint_dual} for $i = i'$.
Then, one can pick up an arbitrary path from vertex $0$ to vertex $i$ on the graph $G$ and update $Y_e \gets Y_e + \epsilon I$ for all $e \in E$ on the path, where $\epsilon > 0$ is an arbitrary constant.
Note that such a path exists by the assumption and can be found by tracing edges backward from vertex $i'$ to vertex $0$.
This solution update results in $\sum_{e \in E} b_{i'e} Y_e = - \epsilon I$,
thus making the constraint~\eqref{eq:constraint_dual} for $i = i'$ satisfied since $- \epsilon I \prec k_{i'} S_{i'}$.
The update does not change the left-hand side of the constraint~\eqref{eq:constraint_dual} for $i \neq i'$ because $\epsilon I$ and $- \epsilon I$ cancel out for vertex $i$ between vertices $0$ and $i'$.
Accordingly, we can obtain a feasible solution by repeating such updates until no more $i'$ violates the constraint.
We start the dual block coordinate ascent algorithm with the obtained feasible solution.

Now, let us consider optimizing $Y_e$ for $e=(i,j)$ such that $i \neq 0$ and $j \neq 0$.
Let
\begin{align}
    \label{eq:def_A_B}
    A := S_i - \frac{1}{k_i} \sum_{e' \in E \setminus \{e\}} b_{ie'} Y_{e'},
    \quad B := S_j - \frac{1}{k_j} \sum_{e' \in E \setminus \{e\}} b_{je'} Y_{e'}.
\end{align}
Then, if we fix all $Y_{e^\prime}$ ($e^\prime \neq e$), the subproblem for $Y_e$ can be written as
\begin{equation}
    \label{sub_opt_problem}
    \max _{Y_e \in \bbS_+^p} k_i \log \det \Big( A - \frac{1}{k_i}Y_e \Big) + k_j \log \det \Big(B + \frac{1}{k_j} Y_e\Big),
    \quad 
    \text{s.t. $- k_j B \prec Y_e \prec k_i A$.}
\end{equation}
Note that if we start with a feasible solution $Y$ for the problem~\eqref{prob:form4}, the feasibility of the subproblem is preserved throughout the optimization procedure.
Note also that if the problem \eqref{sub_opt_problem} is feasible, $A\succ O$ must hold.

\begin{proposition}
    \label{prop:subproblem_optimum_1}
    If the problem \eqref{sub_opt_problem} is feasible, its optimal solution is given by
    \begin{equation*}
        Y_e = \frac{k_ik_j}{k_i+k_j} A^{1/2} \proj_{\bbS_+^p} \Big( I - A^{-1/2} B A^{-1/2} \Big) A^{1/2},
    \end{equation*}
    where $\proj_{\bbS_+^p} (\cdot)$ is a projection onto the set $\bbS_+^p$.
\end{proposition}

\begin{proof}
Let $X:= A^{-1/2}Y_e A^{-1/2}$ and $C:=A^{-1/2}BA^{-1/2}$.
Then the problem~\eqref{sub_opt_problem} is written as
\begin{align*}
    \max_{X\in\bbS_+^p}
    \Big\{ g(X) := k_i \log \det \Big(I-\frac{1}{k_i}X \Big) + k_j \log \det \Big(C + \frac{1}{k_j} X \Big)\Big\}, \quad
    \text{s.t. $- k_j C \prec X \prec k_i I$}
\end{align*}
up to a constant term.
We will show that the solution 
$X^* \coloneqq \frac{k_ik_j}{k_i+k_j} \proj_{\bbS_+^p} (I-C)$
is optimal under the feasibility of the problem, i.e., $- k_j C \prec k_i I$.
First, we see the feasibility of $X^*$ as follows:
\begin{align*}
    X^*
    \succ
    \frac{k_ik_j}{k_i+k_j} \proj_{\bbS_+^p} \! \Big( {-\frac{k_j}{k_i}} C - C \Big)
    =
    \proj_{\bbS_+^p} (- k_j C)
    \succeq
    - k_j C, \quad
    X^*
    \prec
    \frac{k_ik_j}{k_i+k_j} \proj_{\bbS_+^p} \! \Big( I + \frac{k_i}{k_j} I \Big)
    =
    \proj_{\bbS_+^p} (k_i I)
    =
    k_i I.
\end{align*}
Next, since $g$ is a concave function, the optimality of $X^*$ is equivalent to
\begin{align*}
    \langle \nabla g(X^*), X - X^* \rangle \leq 0\quad
    \text{for all $X \in \bbS_+^p$ such that $- k_j C \prec X \prec k_i I$.}
\end{align*}
One sufficient condition is $\nabla g(X^*) \preceq O$ and $\langle \nabla g(X^*), X^*\rangle = 0$, which we can validate by using
\begin{equation*}
    \nabla g(X) = - \Big( I - \frac{1}{k_i}X\Big)^{-1} + \Big( C+\frac{1}{k_j} X \Big)^{-1},
\end{equation*}
and diagonalizing $C$.
\end{proof}

Next, let us consider optimizing $Y_e$ for $e=(i,j)$ such that $i = 0$ and $j \neq 0$.\footnote{Note that $G$ does not have edge $e = (i, j)$ such that $i \neq 0$ and $j = 0$ under the assumption that $G$ is a DAG and that there exists a path from vertex $0 \in V$ to every other vertex.}
Let $B$ be defined by \eqref{eq:def_A_B}, and then the subproblem can be written as
\begin{equation}
    \label{sub_opt_problem_0j}
    \max _{Y_e \in \bbS_+^p}\ 
    - \trace \prn*{Y_e \Gamma^{-1}} 
    + k_j \log \det \Big(B + \frac{1}{k_j} Y_e \Big),
    \quad 
    \text{s.t. $- k_j B \prec Y_e$.}
\end{equation}
We can also write down the optimal solution to this problem as follows.
\begin{proposition}
    The optimal solution to the problem \eqref{sub_opt_problem} is given by
    \begin{equation*}
        Y_e = k_j \Gamma^{1/2} \proj_{\bbS_+^p} \Big( I - \Gamma^{-1/2} B \Gamma^{-1/2} \Big) \Gamma^{1/2}.
    \end{equation*}
\end{proposition}
We omit the proof because the proof is similar to that of Proposition~\ref{prop:subproblem_optimum_1}.

The overall algorithm is shown in Algorithm~\ref{alg:dca1}.
This algorithm can be directly used to solve \eqref{form2}.

\begin{algorithm}
\caption{Dual block coordinate ascent method}\label{alg:dca1}
\begin{algorithmic}[1]
\Require $G = (V, E)$, $(k_i)_{i \in V \setminus \set{0}}$, $(S_i)_{i \in V \setminus \set{0}}$, $\epsilon > 0$
\For {$e\in E$} 
    \Comment{Initialization}
    \State $Y_e \leftarrow O\in \bbR^{p\times p}$
\EndFor
\For {$i \in V \setminus \set{0} $ such that $S_i$ is singular} 
    \Comment{Making $(Y_e)_{e \in E}$ feasible for the dual problem}
    \State Pick up an arbitrary path from vertex $0$ to vertex $i$ on $G$ and update $Y_e \gets Y_e + \epsilon I$ for all $e \in E$ on the path
\EndFor
\Repeat
    \Comment{Main loop}
    \For {$e=(i,j)\in E$ such that $i = 0$}
        \State Set $B$ as in \eqref{eq:def_A_B}
        \State $Y_e \gets k_j \Gamma^{1/2} \proj_{\bbS_+^p} \Big( I - \Gamma^{-1/2} B \Gamma^{-1/2} \Big) \Gamma^{1/2}$
        \Comment{Diagonalization is useful}
    \EndFor
    \For {$e=(i,j)\in E$ such that $i \neq 0$}
        \State Set $A$ and $B$ as in \eqref{eq:def_A_B}
        \State $Y_e \gets \frac{k_ik_j}{k_i+k_j} A^{1/2} \proj_{\bbS_+^p} \Big( I - A^{-1/2} B A^{-1/2} \Big) A^{1/2}$
    \EndFor
\Until{convergence}
\For {$i \in V \setminus \set{0}$} 
    \Comment{Recovering $Q_i$ from $(Y_e)_{e\in E}$}
    \State $Q_i \leftarrow S_i - \frac{1}{k_i} \sum_{e \in E} b_{ie} Y_e$
\EndFor \\
\Return{$(Q_i)_{i\in V}$}
\end{algorithmic}
\end{algorithm}

\section{Numerical tests}

We test the new discretization error model using the algorithm presented above.
Due to the space constraints, we restrict our discussion to results from the discretization error quantification.
To focus on the $\Sigma$-estimation process, we hold the unknowns, $\theta$, in the model differential equation to their true values.
We employ artificially generated observations for the experiments.
These settings are highly unrealistic but still provide a viable framework for probing the characteristics of our proposed approach.

As a toy problem, we consider the Lorenz system
\begin{equation*}
    \frac{\rmd}{\rmd t}
    \begin{bmatrix}
        x_1 \\ x_2 \\ x_3
    \end{bmatrix}
    =
    \begin{bmatrix}
        \sigma (-x_1 + x_2) \\
        x_1 (\rho - x_3) - x_2 \\
        x_1 x_2 - \beta x_3
    \end{bmatrix},
    \quad 
    \begin{bmatrix}
        x_1(0) \\ x_2 (0) \\ x_3 (0)
    \end{bmatrix}
    =
    \begin{bmatrix}
        -10 \\ -1 \\ 40
    \end{bmatrix},
\end{equation*}
where $(\sigma,\rho,\beta) = (10,28,8/3)$.
We designate the observation operator as $\calO(x) = (x_1,x_2,x_3)^\trans$.
The observation noise variance is set to $\Gamma = \mathrm{diag}(0.05^2,0.01^2,0.05^2)$, with observations assumed to be obtained at $t_i = (i-1)h$ with $i=1,2,\dots,300$ and $h=0.05$, indicating $t\in [0,15]$.
In the numerical example, the degrees of freedom for the Wishart distribution is set to 3, i.e. $k_i =3$ in \eqref{eq:partition}.

Fig.~\ref{fig:deq_kepler} illustrates the discretization error quantification results.
By definition,
each estimated $\Sigma_i$ is a $3\times 3$ matrix. 
We projected these results into two dimensions and visualize them by drawing ellipses. 
At each point in time, a pair of ellipses corresponds to the probabilities of $68\%$ and $95\%$.
The results show that a significant correlation between $x_1 $ and $x_2$ is captured, although it may vary significantly as time passes.
The actual errors are also depicted in these figures, with some appearing slightly outside the outer ($95\%$) ellipse.
This behaviour often happens when the error grows sharply, and similar behaviour was reported in our previous report~\cite{mm21}.
Table~\ref{tab:frequency} demonstrates the frequency at which the actual error is encompassed within the ellipses corresponding to $68\%$ and $95\%$ probabilities. 
Upon conducting similar experiments with varying parameters, we observed a pattern: $70$-$90\%$ of the actual errors were typically embraced within the ellipse of $95\%$ probability, and the results for $(x_2,x_3)$ and $(x_3,x_1)$ were almost the same.
Besides, a comparison of ellipses at different time points reveals that the ellipse associated with a larger $t$ value embraces the one with a smaller $t$, indicating that the algorithm preserves the monotonicity constraint.

In future publications, we plan to provide comprehensive applications and detailed analyses for more practical inverse problems.

\begin{figure}[t]
\centering
\begin{tikzpicture}[scale=0.78]
\begin{axis}[
    xmin=-75,xmax=75,
    ymin=-75,ymax=75,
    xtick={-60,-40,...,60},
    ytick={-60,-40,...,60},
    xlabel={error in $x_1$},
    ylabel={error in $x_2$},
    ]
    
    \def\a{7.387050564793316}
    \def\b{2.6662008433739066}
    \def\angle{61.08873591353575}
    \addplot [very thick, color=mycolor1, domain=-360:360, samples=600] ({\a*cos(\x)*cos(\angle) -\b* sin(\x)*sin(\angle)}, {\a*cos(\x)*sin(\angle) + \b*sin(\x)*cos(\angle)});

    \def\a{11.028827220005853}
    \def\b{3.980623684309035}
    \addplot [very thick, color=mycolor1, domain=-360:360, samples=600] ({\a*cos(\x)*cos(\angle) -\b* sin(\x)*sin(\angle)}, {\a*cos(\x)*sin(\angle) + \b*sin(\x)*cos(\angle)});

    \addplot[mycolor1,mark=o,only marks] coordinates{
        (-1.5963346994043484, -2.2895965421164366)
        (-1.637422625617937, -0.9193189062429088)
        (-0.7911590026520017, 1.7894463975797326)};

    \def\a{59.93829173507859}
    \def\b{9.873318126080575}
    \def\angle{50.82380194080455}
    \addplot [very thick, color=mycolor2, domain=-360:360, samples=600] ({\a*cos(\x)*cos(\angle) -\b* sin(\x)*sin(\angle)}, {\a*cos(\x)*sin(\angle) + \b*sin(\x)*cos(\angle)});

    \def\a{89.48755089873713}
    \def\b{14.74081147077424}
    \addplot [very thick, color=mycolor2, domain=-360:360, samples=600] ({\a*cos(\x)*cos(\angle) -\b* sin(\x)*sin(\angle)}, {\a*cos(\x)*sin(\angle) + \b*sin(\x)*cos(\angle)});

    \addplot[mycolor2,mark=o,only marks] coordinates{
        (18.15089280518594, 25.97694545324728)
        (21.37455715112562, 26.01610667603766)
        (22.639208409755064, 23.21395410576235)};
\end{axis}
\end{tikzpicture}
\begin{tikzpicture}[scale=0.78]
\begin{axis}[
    xmin=-75,xmax=75,
    ymin=-75,ymax=75,
    xtick={-60,-40,...,60},
    ytick={-60,-40,...,60},
    xlabel={error in $x_2$},
    ylabel={error in $x_3$},
    ]
    
    \def\a{7.879031133705079}
    \def\b{4.472641872654493}
    \def\angle{72.49054692536554}
    \addplot [very thick, color=mycolor1, domain=-360:360, samples=600] ({\a*cos(\x)*cos(\angle) -\b* sin(\x)*sin(\angle)}, {\a*cos(\x)*sin(\angle) + \b*sin(\x)*cos(\angle)});

    \def\a{11.763351593778005}
    \def\b{6.67763053708703}
    \addplot [very thick, color=mycolor1, domain=-360:360, samples=600] ({\a*cos(\x)*cos(\angle) -\b* sin(\x)*sin(\angle)}, {\a*cos(\x)*sin(\angle) + \b*sin(\x)*cos(\angle)});

    \addplot[mycolor1,mark=o,only marks] coordinates{
        (-2.2895965421164366, -1.634741754586809)
        (-0.9193189062429088, -3.597780757592851)   
        (1.7894463975797326, -3.7851116689144675)};

    \def\a{32.82085713642936}
    \def\b{15.31746295079621}
    \def\angle{-88.88975250567525}
    \addplot [very thick, color=mycolor2, domain=-360:360, samples=600] ({\a*cos(\x)*cos(\angle) -\b* sin(\x)*sin(\angle)}, {\a*cos(\x)*sin(\angle) + \b*sin(\x)*cos(\angle)});

    \def\a{49.00136521270765}
    \def\b{22.868890750296234}
    \addplot [very thick, color=mycolor2, domain=-360:360, samples=600] ({\a*cos(\x)*cos(\angle) -\b* sin(\x)*sin(\angle)}, {\a*cos(\x)*sin(\angle) + \b*sin(\x)*cos(\angle)});

    \addplot[mycolor2,mark=o,only marks] coordinates{
        (25.97694545324728, 13.25769850462524)
        (26.01610667603766, 19.343438083867966)
        (23.21395410576235, 18.577569877989973)};
\end{axis}
\end{tikzpicture}

\begin{tikzpicture}[scale=0.78]
\begin{axis}[
    xmin=-75,xmax=75,
    ymin=-75,ymax=75,
    xtick={-60,-40,...,60},
    ytick={-60,-40,...,60},
    xlabel={error in $x_3$},
    ylabel={error in $x_1$},
    ]
    
    \def\a{7.982077770799678}
    \def\b{2.9417128394256356}
    \def\angle{-10.758884509532049}
    \addplot [very thick, color=mycolor1, domain=-360:360, samples=600] ({\a*cos(\x)*cos(\angle) -\b* sin(\x)*sin(\angle)}, {\a*cos(\x)*sin(\angle) + \b*sin(\x)*cos(\angle)});

    \def\a{11.917199675113125}
    \def\b{4.391961629655625}
    \addplot [very thick, color=mycolor1, domain=-360:360, samples=600] ({\a*cos(\x)*cos(\angle) -\b* sin(\x)*sin(\angle)}, {\a*cos(\x)*sin(\angle) + \b*sin(\x)*cos(\angle)});

    \addplot[mycolor1,mark=o,only marks] coordinates{
        (-1.634741754586809, -1.5963346994043484)
        (-3.597780757592851, -1.637422625617937)
        (-3.7851116689144675, -0.7911590026520017)};

    \def\a{33.566440213019476}
    \def\b{12.581084936705889}
    \def\angle{5.047485192908817}
    \addplot [very thick, color=mycolor2, domain=-360:360, samples=600] ({\a*cos(\x)*cos(\angle) -\b* sin(\x)*sin(\angle)}, {\a*cos(\x)*sin(\angle) + \b*sin(\x)*cos(\angle)});

    \def\a{50.11451678216665}
    \def\b{18.783492923204296}
    \addplot [very thick, color=mycolor2, domain=-360:360, samples=600] ({\a*cos(\x)*cos(\angle) -\b* sin(\x)*sin(\angle)}, {\a*cos(\x)*sin(\angle) + \b*sin(\x)*cos(\angle)});

    \addplot[mycolor2,mark=o,only marks] coordinates{
        (13.25769850462524, 18.15089280518594)
        (19.343438083867966, 21.37455715112562)
        (18.577569877989973, 22.639208409755064)};
\end{axis}
\end{tikzpicture}

\caption{Discretization error quantification results.
The blue markers represent results for time interval $t\in[7.4,7.5]$, corresponding to $\tilde{\Sigma}_{150}$. 
Each plot features two ellipses, representing the $68\%$ and $95\%$ probability levels, respectively. Within the $t\in[7.4,7.5]$ range, there are three discrete time points, with actual errors displayed as circles, which are obtained by comparing the numerical approximation to the reference solution.
The red markers indicate results for the time interval $t\in[14.9,15.0]$
}
\label{fig:deq_kepler}
\end{figure}
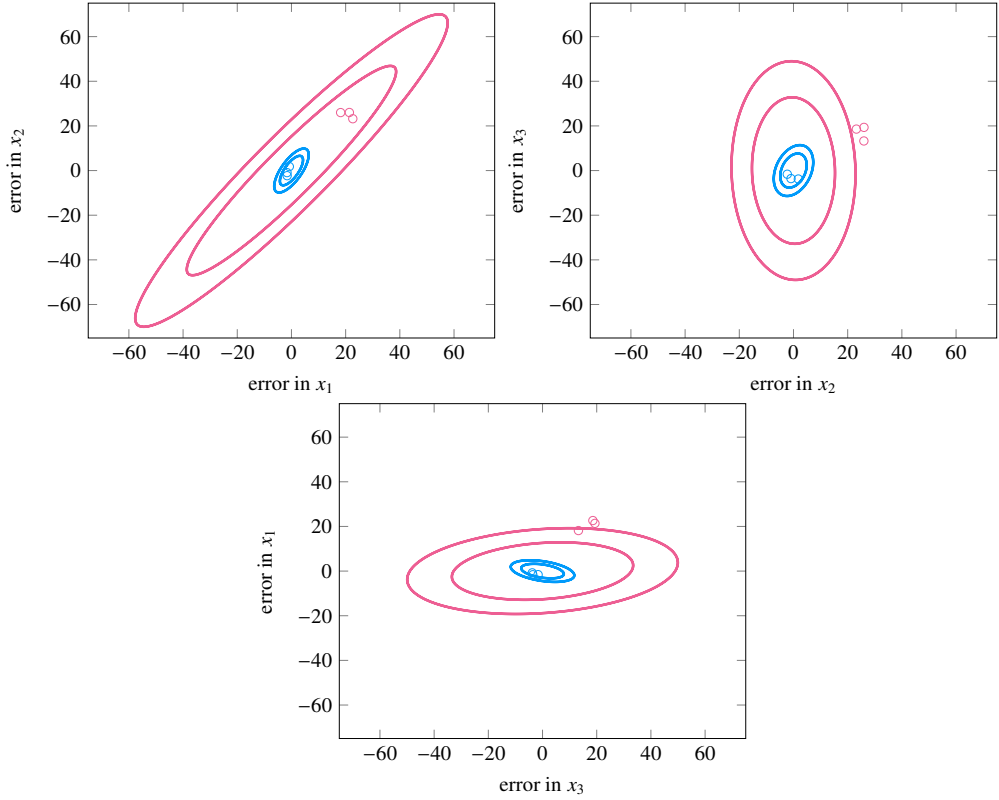

\begin{table}[t]
    \caption{Frequency that the actual error lies in the ellipses corresponding to the $68\%$ and $95\%$ probabilities.
    We note that the results for $i=1,\dots,18$ are disregarded because the estimated $\tilde{\Sigma}$'s were singular.}
    \label{tab:frequency}
    \centering
    \begin{tabular}{c|ccc}
         &  $(x_1,x_2)$ & $(x_2,x_3)$ & $(x_3,x_1)$\\
         \hline
         $68\%$ & $80.0\%$ & $59.3\%$ & $59.8\%$ \\
         $95\%$ & $86.6\%$ & $73.2\%$ & $73.2\%$
    \end{tabular}
\end{table}



 \bibliographystyle{elsarticle-num} 
 \bibliography{ref}





\end{document}